\documentclass{article}
\usepackage{ijcai18}

\usepackage{times}
\usepackage{xcolor}
\usepackage{soul}
\usepackage[utf8]{inputenc}
\usepackage[small]{caption}
\usepackage[T1]{fontenc}
\usepackage[nosumlimits,nonamelimits]{amsmath}
\usepackage{dsfont}

\usepackage{booktabs}

\usepackage{mathtools}

\usepackage{algorithmic}
\usepackage{algorithm}

\usepackage[caption=false,font=footnotesize]{subfig}

\usepackage{amssymb,amsmath,amsthm}
\usepackage{thmtools}
\usepackage{thm-restate}

\newtheorem{proposition}{Proposition}
\newtheorem{property}{Property}
\newtheorem{observation}{Observation}
\newtheorem{lemma}{Lemma}

\DeclarePairedDelimiter\floor{\lfloor}{\rfloor}
\newcommand{\deps}{\Delta\epsilon}

\newcommand{\dc}{\Delta c}
\usepackage{booktabs} 
\newlength{\textfloatsepsave} 

\title{Designing the Game to Play: Optimizing Payoff Structure in Security Games}

\author{
Zheyuan Ryan Shi$^{1*}$, 
Ziye Tang$^{2}$\thanks{Z. R. Shi and Z. Tang contributed equally to this work.}, 
Long Tran-Thanh$^3$, 
Rohit Singh$^4$,
Fei Fang$^2$
\\ 
$^1$ Swarthmore College, USA \\
$^2$ Carnegie Mellon University, USA\\
$^3$ University of Southampton, UK\\
$^4$ World Wide Fund for Nature, Cambodia\\
zshi1@swarthmore.edu,
ziyet@andrew.cmu.edu,\\
ltt08r@ecs.soton.ac.uk, 
rsingh@wwfnet.org,
feifang@cmu.edu
}

\begin{document}

\maketitle

\begin{abstract}
We study Stackelberg Security Games where the defender, in addition to allocating defensive resources to protect targets from the attacker, can strategically manipulate the attacker's payoff under budget constraints in weighted $L^p$-norm form regarding the amount of change. 
For the case of weighted $L^1$-norm constraint, we present (i) a mixed integer linear program-based algorithm with approximation guarantee; (ii) a branch-and-bound based algorithm with improved efficiency achieved by effective pruning; (iii) a polynomial time approximation scheme for a special but practical class of problems. 
In addition, we show that problems under budget constraints in $L^0$ and weighted $L^\infty$-norm form can be solved in polynomial time. 
\end{abstract}

\section{Introduction}

Research efforts in security games have led to success in various domains, ranging from protecting critical infrastructure~\cite{Letchford:2013:AAAI:marginal,Wang:2016:AAAI:terrorist} and catching fare invaders in metro systems~\cite{Yin:2012:IAAI:trusts}, to combating poaching~\cite{Fang:2016:AAAI:paws} and preventing cyber intrusions~\cite{Durkota:2015:IJCAI:netsec,Basilico:2016:ARES:software}. In these games, a defender protects a set of targets from an attacker by allocating defensive resources. One key element that characterizes the strategies of the players is the payoff structure. Existing work in this area typically treats the payoff structure of the players as given parameters, sometimes with uncertainties known a priori given the nature of the domain. 
However, under various circumstances, the defender is able to change the attacker's payoff, thus rendering the existing models inadequate in expressiveness. 
For example, in wildlife poaching, the law enforcement agency may charge a variable fine if the poacher is caught at different locations, e.g., in the core area vs.\ in the non-core area. 
In cybersecurity, the network administrator may change the actual or appeared value of any network node for a potential attacker.
In these cases, the defender's decision making is two-staged: she chooses the payoff structure, as well as the strategy of allocating defensive resources. With a properly chosen payoff structure, the defender may save the effort by achieving much better utility with the same or even less amount of resources.


As existing work in security games does not provide adequate tools to deal with this problem (see Section 2 for more details), we aim to fill this gap as follows.
%
%
We study how to design the attacker's payoff structure in security games given budget constraints in weighted $L^p$-norm ($p=0,1,\infty$). That is, the distance between the original payoff structure and the modified payoff structure is bounded, using distance metrics such as Manhattan distance (i.e., $L^1$-norm) with varying weights for reward and penalty of different targets.
The intuition behind this setting is that the defender can change the payoffs to make a target that is preferable to the defender more attractive to the attacker and disincentivize the attacker from attacking targets that can lead to a significant loss to the defender.
More change incurs a higher cost to the defender and the defender has a fixed budget for making the changes.
Our findings can be summarized as follows:

\textbf{$L^1$-norm case}: When the budget constraint is in weighted $L^1$-norm form, i.e.\ additive cost, our contribution is threefold. (i) We exploit several key properties of the optimal manipulation and propose a \emph{mixed integer linear program} (MILP)-based algorithm with \emph{approximation guarantee}.
(ii) We propose a \emph{novel branch-and-bound approach} with improved efficiency achieved by effective pruning for the general case.
(iii) Finally, we show that a \emph{polynomial time approximation scheme} (PTAS) exists for a special but practical case where the budget is very limited and the manipulation cost is uniform across targets. The PTAS is built upon the key observation that there is an optimal solution where no more than two targets' payoffs are changed in this restricted case.

\textbf{$L^0$ and $L^\infty$-norm cases}: We propose \emph{a $O(n^3)$ and a $O(n^2\log n)$ algorithm} for problems under budget constraints in $L^0$-norm form and weighted $L^\infty$-norm form\footnote{The $L^0$-norm is not actually a norm, but we use the term for simplicity and its definition is given in Sec.~\ref{subsec:l0}. The definition of weighted $L^\infty$-norm is given in Sec.~\ref{subsec:linfty}.}, respectively, where $n$ is the total number of targets. For $L^0$-norm form budget, i.e.\ limited number of targets to manipulate, our algorithm converts the problem into $O(n^2)$ subproblems and reduces each subproblem to a problem of finding a subset of items with the maximum average weight. The latter can be solved in $O(n)$ time. For $L^\infty$-norm form budget, i.e.\ limited range of manipulation on each target, we reduce the problem to traditional Stackelberg Security Games with fixed payoff structure, which again admits an efficient algorithm.

\textbf{Numerical evaluation}: We provide extensive experimental evaluation for the proposed algorithms. For problems with $L^1$-norm form budget constraint, we show that the branch-and-bound approach with an additive approximation guarantee can solve up to hundreds of targets in a few minutes. This is faster than other baseline algorithms we compare to. Somewhat surprisingly, naively solving $n$ non-convex subproblems using interior point method achieves good performance in practice
despite that
there is no theoretical guarantee of solution quality. We also evaluate the proposed $O(n^3)$ algorithm for the $L^0$-norm form case and show its superior performance over two greedy algorithms and a MILP based algorithm.





\section{Preliminaries and Related Work}
The security game that we consider in this paper features a set of $n$ targets, $T = \{1,2,\dots, n\}$. The defender has $r$
units of defensive resources, each can protect one target. The attacker can choose to attack one target after observing the defender's strategy. 
If the defender covers target $i$ when it is attacked, the defender gets a reward $R^d_i$ $\geq 0$ and the attacker gets a penalty $P^a_i$ $\leq 0$. Otherwise, the defender gets a penalty $P^d_i$ $\leq 0$ and the attacker gets a reward $R^a_i$ $\geq 0$. When the defender commits to a mixed strategy $c$, that is, covering target $i$ with probability $c_i$, the defender's and attacker's expected utilities when target $i$ is attacked are $U^d_i = c_i R^d_i + (1-c_i) P^d_i$ and $U^a_i = c_i P^a_i + (1-c_i) R^a_i$, respectively.

We adopt the commonly used solution concept of Strong Stackelberg Equilibrium (SSE) \cite{Kiekintveld:2009:AAMAS:origami}. At an SSE, the defender chooses an optimal strategy that leads to the highest expected utility for her when the attacker chooses a best response (assumed to be a pure strategy w.l.o.g), breaking ties in favor of the defender. Given a coverage $c$, the attack set $\Gamma \subseteq T$ contains all the targets which have a weakly higher attacker's expected utility than any other target, i.e.,
\begin{equation} \label{equ:attackset}
\Gamma = \{j \in T: U^a_j \geq U^a_k, \forall k \in T\}
\end{equation} 
\cite{Kiekintveld:2009:AAMAS:origami} show that there exists an SSE where the defender only covers the targets in the attack set. 

Given the game parameters, the optimal defender strategy in such a game can be computed using multiple linear programs (LPs)~\cite{Conitzer:2006:EC:multipleLP} or an efficient $O(n^2)$ algorithm called ORIGAMI \cite{Kiekintveld:2009:AAMAS:origami} based on enumerating the possible attack sets. We leverage insights from both works to devise our algorithms. 

Although many algorithms have been developed for security games under various settings, in most of the existing literature, the payoff structure is treated as fixed and cannot be changed by the defender, either in the full information case
~\cite{Korzhyk:2010:AAAI:multipleresources,Paruchuri:2008:AAMAS:bsg,Laszka:2017:AAAIAICS:alert}, or in the presence of payoff uncertainties~\cite{Kiekintveld:2013:AAMAS:interval,Kiekintveld:2011:AAMAS:bayesian,Yin:2012:AAMAS:bayesian,Letchford:2009:SAGT:learning,Blum:2014:NIPS:learning}. 
%
%
%
As mentioned earlier, in many real-world scenarios the defender has control over the attacker's payoffs.
The approaches above ignore this aspect and thus leave room for further optimization. 

Indeed, despite its significance, jointly optimizing the payoff structure and the resource allocation is yet under-explored. 
A notable exception, and a most directly related work to ours, is the audit game model~\cite{Blocki:2013:IJCAI:audit1,Blocki:2015:AAAI:audit2}. The defender can choose target-specific ``punishment rates'', in order to maximize her expected utility offset by the cost of setting the punishment rate. 
Compared with their model, ours is more general in that we allow not only manipulation of attacker's penalty, but also attacker's reward. This realistic extension makes their core techniques inapplicable. Also, we treat the manipulation cost as a constraint instead of a regularization term in the objective function, for in some real-world settings, payoffs can be manipulated only once, yet the defender may face multiple attacks afterwards. This makes it hard to determine the regularization coefficient.
Another closely related work~\cite{Schlenker:2018:AAMAS:cyber} 
focuses
on the use of honeypot~\cite{Kiekintveld:2015:CyberWarfare:honeypot,Durkota:2015:IJCAI:netsec,Pibil:2012:GameSec:honeypot}.
It studies the problem of deceiving a cyber attacker by manipulating the attacker's (believed) payoff. However, it assumes the defender can only change the payoff structure, ignoring the allocation of defensive resources after the manipulation. ~\cite{Horak:2017:GameSec:belief} study the manipulation of attacker's belief in a repeated game. They assume actively engaging attacker and defender, which is not the case in our problem.

If we conceptually decouple the payoff manipulation from resource allocation, the defender faces a two-stage decision. She first chooses the structure of the game, and then plays the game. Thus, our problem may be viewed as a mechanism design problem, albeit not in a conventional setting. Most work in mechanism design considers private information games~\cite{Fujishima:1999:IJCAI:auction,Myerson:1989:TNP:mechanism}, while in our work, and in most security game literature, the payoff information is public. Some design the incentive mechanism using a Stackelberg game~\cite{Kang:2015:MC:p2pmechanism}, with applications to network routing~\cite{Sharma:2007:EC:routingmechanism}, mobile phone sensing~\cite{Yang:2012:Mobicom:mobilephone}, and ecology surveillance~\cite{Xue:2016:CP:ecology}. However, these works solve the Stackelberg game to design the mechanism, rather than designing the structure of the Stackelberg game.

\section{Optimizing Payoff with Budget Constraint in Weighted $L^1$-norm Form} 
\label{sect:l1}
In this section, we focus on computing the optimal way of manipulating attacker's payoffs and allocating defensive resources when the defender can change the attacker's reward and penalty at a cost that grows linearly in the amount of change and the defender has a limited budget for making the changes. The cost rate, referred to as weights, may be different across targets. This is an abstraction of several domains. For example, a network administrator may change the actual or appeared value of any network node although such change often incurs time and hardware costs. 



Let $R^a$, $P^a$, $\bar{R}^a$, $\bar{P}^a$ denote the attacker's reward and penalty vectors before and after the manipulation. Similar to the initial payoff structure, we require that $\bar{R}^a\geq 0 \geq \bar{P}^a$ and denote $D_j = R^a_j - P^a_j$. Let $\epsilon=\bar{R}^a-R^a$ and $\delta=\bar{P}^a-P^a$ be the amount of change in attacker's reward and penalty and $\mu, \theta$ the weights on $\epsilon, \delta$ resp.. The budget constraint is in weighted $L^1$ norm form, i.e., $\sum_j(\mu_j|\epsilon_j|+\theta_j|\delta_j|)\leq B$ where $B$ is the budget. The defender's strategy is characterized by $(c, \epsilon,\delta)$. Given this strategy, in the manipulated game the attacker attacks some target $t$, which belongs to the attack set $\Gamma$.
We first show some properties of the optimal solution.
\begin{restatable}{thm}{optproperty}
\label{property:l1opt}
    There is an optimal solution $(c,\epsilon,\delta)$ with corresponding attack target $t$ and attack set $\Gamma$ which satisfies the following conditions: 
    \begin{enumerate}
        \item $c_j = 0, \epsilon_j = 0, \delta_j = 0, \forall j \notin \Gamma$.
        \item $\epsilon_t \geq 0, \delta_t \geq 0$; $\epsilon_j \leq 0, \delta_j \leq 0$, $\forall j\neq t$.
        \item $\delta_t \epsilon_t (\delta_t+P^a_t)=0$ and  $\delta_j \epsilon_j (R^a_j+\epsilon_j)=0$, $\forall j\neq t$.
    \end{enumerate}
\end{restatable}
\begin{proof}[Proof sketch]
Condition 1: If any $\epsilon_j, \delta_j \neq 0$ with $j \notin \Gamma$, we may either set $\epsilon_j, \delta_j = 0$ or push $j$ into the attack set. There is no need to protect a target that is not in the attack set. Condition 2: Flipping the sign of $\epsilon$ and $\delta$ leads to no better solution.
Condition 3: 
For each target $i$, we can shift the change on $R^a_i$ to $P^a_i$ and vice versa while one is more budget efficient than the other depending on coverage.
We can show that these manipulations can be done simultaneously.\footnote{Due to limited space, the full proofs are included in the full version of the paper: https://arxiv.org/abs/1805.01987}
\end{proof}

Similar to the multiple LPs formulation in \cite{Conitzer:2006:EC:multipleLP}, we consider $n$ subproblems $\mathcal P_i$, each assuming some target $i \in T$ is the attack target, and the best solution among all $n$ subproblems is the optimal defender strategy. Condition 2 in Property \ref{property:l1opt} shows it is possible to infer the sign of $\epsilon$ and $\delta$ given the attack target. So in the sequel, we abuse the notation by treating $\epsilon, \delta$ as the absolute value of the amount of change, and assume w.l.o.g. that in $\mathcal P_i$, $\bar{R}^a_i = R^a_i + \epsilon_i, \bar{P}^a_i = P^a_i + \delta_i$ and $\forall j\neq i, \bar{R}^a_j = R^a_j - \epsilon_j, \bar{P}^a_j = P^a_j - \delta_j$. Thus, a straightforward formulation for $\mathcal P_i$ is
\begin{align}\label{program:l1cont}
\max\limits_{c,\epsilon,\delta} \quad & U^d_i = R^d_ic_i+P^d_i(1-c_i) = c_i D_i + P_i^d \\
\mbox{s.t.}\quad & 
U^a_i = c_i (P^a_i+\delta_i) + (1-c_i) (R^a_i+\epsilon_i) \label{constr:attack-target}\\
& \nonumber \geq U^a_j = c_j (P^a_j-\delta_j) + (1-c_j) (R^a_j-\epsilon_j), \forall j \neq i\\
& \sum_j(\mu_j\epsilon_j+\theta_j\delta_j)\leq B \label{constr:budget}\\
& \sum_jc_j \leq r \\
&R_j^a-\epsilon_j \geq 0,\quad \forall j\neq i\\
&P_i^a+\delta_i \leq 0\\
& c_j, \epsilon_j, \delta_j \geq 0,\quad c_j \leq 1, \quad \forall j \in T\label{constr:nonnegative}
\end{align}
The above formulation is non-convex due to the quadratic terms in Constraint \ref{constr:attack-target} which leads to an indefinite Hessian matrix (see Appendix~\ref{pf:l1_non_convex}), and thus no existing solvers can guarantee global optimality for the above formulation.


\subsection{A MILP-based Solution with Approximation Guarantee}
\label{subsect:l1atomic}

To find a defender strategy with solution quality guarantee, we solve the atomic version of the subproblems with MILPs. We show an approximation guarantee which improves as the fineness of discretization grows. We further propose a branch-and-bound-like framework for pruning subproblems to improve runtime efficiency. 

In the atomic version of the payoff manipulation problem, we assume the defender can only make atomic changes, with the minimum amount of change given as $\rho_0$.
We refer to the atomic version of $\mathcal{P}_i$ as $\mathcal{AP}_i$. 
$\mathcal{AP}_i$ can be formulated as the MILP in Equations~\ref{eqn:l1milpbegin}-\ref{eqn:APi-linearize2}. We simplify the objective function as $c_i$ since $D_i\geq 0$. All constraints involving sub/super-script $j, k$ without a summation apply to all proper range of summation indices. We use binary representation for $\bar R^a_i/\rho_0$ and $\bar P^a_i/\rho_0$ in constraints~\ref{eqn:APi-binary1}-\ref{eqn:APi-binary5}. 
The binary representation results in bilinear terms like $y^k_jc_j$. We introduce variables $\alpha^k_j, \beta^k_j$ and constraints~\ref{eqn:APi-linearize1}-\ref{eqn:APi-linearize2} to linearize them.
\begin{align}
\max\limits_{y_j^k, z_j^k, c_j} \quad & c_i \label{eqn:l1milpbegin}\\
s.t. \quad \nonumber &\text{Constraint }\ref{constr:budget}\text{-}\ref{constr:nonnegative}\\
&\epsilon_i = \rho_0\sum_k 2^k y^k_i-R^a_i\label{eqn:APi-binary1}\\
&\epsilon_j = R^a_j - \rho_0\sum_k 2^k y^k_j, \quad \forall j\neq i\\
&\delta_i=-\rho_0\sum_k 2^k z^k_i-P^a_i\\
&\delta_j=P^a_j + \rho_0\sum_k 2^k z^k_j,\quad \forall j\neq i\label{eqn:APi-binary4}\\
&y^k_j, z^k_j \in \{0,1\} \label{eqn:APi-binary5}\\
&v_i\geq v_j \\
&v_i = R^a_i + \epsilon_i - \rho_0\sum_k 2^k(\alpha^k_i + \beta^k_i )\\ 
&v_j = R^a_j - \epsilon_j - \rho_0\sum_k 2^k(\alpha^k_j + \beta^k_j ),\forall j\neq i\\
&0\leq \alpha^k_j\leq y^k_j,\quad c_j -(1-y^k_j)\leq \alpha^k_j\leq c_j \label{eqn:APi-linearize1}\\
&0\leq \beta^k_j\leq z^k_j,\quad c_j-(1-z^k_j)\leq \beta^k_j\leq c_j\label{eqn:APi-linearize2}
\end{align}

The optimal defender strategy for the atomic payoff manipulation problem can be found by checking the solution to all the subproblems and compare the corresponding $U_i^d$. We can also combine all the subproblems by constructing a single MILP, with additional variables indicating which subproblem is optimal. The details can be found in the full version.

A natural idea to approximate the global optima of the original $L^1$-constrained payoff manipulation problem is, for each attack target $i$, approximate $\mathcal P_i$ with $\mathcal{AP}_i$ using small enough $\rho_0$. Theorem~\ref{thm:l1milpapprox} below shows such an approximation bound.
\begin{restatable}{thm}{milpapprox}
\label{thm:l1milpapprox}
The solution of the atomic problem is an additive 
$\max_i \frac{2\rho_0(R^d_i - P^d_i)}{R^a_i}$-approximation to the original problem.
\end{restatable}
\begin{proof}[Proof sketch]
The floor and ceiling notations are about the ``integral grid'' defined by $\rho_0$. Suppose $(c^*, \epsilon^*,\delta^*)$ is an optimal solution to $\mathcal P_i$. Let $\epsilon' = \lfloor \epsilon^* \rfloor$, $\delta' = \lfloor \delta^* \rfloor$, and $c' = c^*$ except $c_i' = c_i^* - 2\rho_0/(D_i + \epsilon_i' - \delta_i')$. We can show such feasible solutions yield the desired approximation bound.
\end{proof}

We note that the idea of discretizing the manipulation space is similar to~\cite{Blocki:2013:IJCAI:audit1,Blocki:2015:AAAI:audit2}. Yet allowing changes in both reward and penalty and the difference in objective function make our formulation different and the reduction to SOCP used in ~\cite{Blocki:2015:AAAI:audit2} inapplicable.

We can further improve the practical runtime of the MILPs by pruning and prioritizing subproblems as shown in Alg.~\ref{alg:bnb}. 
We first compute a global lower bound by checking a sequence of greedy manipulations. Inspired by Condition 2 and 3 in Property \ref{property:l1opt}, we greedily spend all the budget on one target to increase its reward or penalty, leaving all other targets' payoff parameters unchanged (Lines \ref{ln:GM_start} - \ref{ln:GM_end}).

Upper bounds in $\mathcal{P}_i$ can be computed with budget reuse: we independently spend the full amount of budget $B$ on each target to increase $R^a_i$ and $P^a_i$ and decrease $R^a_j$ and $P^a_j$, $j \neq i$, as much as possible.
%
For the ease of notation, in Alg.~\ref{algo:l1-pruning} we assume manipulations have uniform cost. The weighted case can be easily extended.

The subproblem $\mathcal{P}_i$ is pruned if its upper bound is lower than the global lower bound. To make the pruning more efficient, we solve subproblems in descending order of their corresponding lower bounds, hoping for an increase in the global lower bound. For subproblems that cannot be pruned, we set $\rho_0$ to the desired accuracy and solve the MILP to approximate the subproblem optima. We also add to the MILP the linear constraint on $c_i$ derived from the global lower bound.

To get the bounds, we call an improved version of the ORIGAMI algorithm in~\cite{Kiekintveld:2009:AAMAS:origami} by doing a binary search on the size of the attack set $\Gamma$, and solve the linear system. It is denoted as ORIGAMI-BS in Alg.~\ref{alg:bnb}. Recall $r$ is the defender's total resource.
Let $M$ be the attacker's expected utility for attack set and $\bar E_k = \frac{1}{\bar R^a_i - \bar P^a_i}$. From $U^a_i = U^a_j,\forall j\in\Gamma$ and $\sum_{j \in \Gamma} c_j = r$, we obtain
\begin{align} 
M &= \frac{\sum_{k \in \Gamma} \bar{R}^a_k \bar E_k - r}{\sum_{k \in \Gamma}\bar E_k} & \label{eqn:M}\\
c_j &= \bar E_j \left( \frac{\sum_{k \in \Gamma} (\bar R^a_j - \bar R^a_k) \bar E_k + r}{\sum_{k \in \Gamma} \bar E_k}\right),  & \quad \forall j\in \Gamma \label{eqn:cj}
\end{align}
We iteratively cut the search space by half based on $c_j$ and $M$. The complexity improves from $O(n^2)$ to $O(n\log n)$. A complete description can be found in the full version.

\setlength{\textfloatsep}{5pt}
\begin{algorithm}[t]
\scriptsize
    \caption{Branch-and-bound}    \label{alg:bnb}     
    \begin{algorithmic}[1] \label{algo:l1-pruning}
        \REQUIRE Payoffs $\sigma = \{R^d, P^d, R^a, P^a\}$, budget $B$
        \STATE Initialize $LB\leftarrow \emptyset, globalLB \leftarrow -\infty, N \leftarrow \emptyset$ containing set of indices of pruned subproblems. Set $\rho_0$ to be a desired accuracy.
        \FOR{Subproblem $\mathcal P_i$}\label{ln:GM_start}
        \STATE Greedy Modifications (GM):
        \STATE $\mbox{GM}_1 \leftarrow \bar{R}^a_i=R^a_i+B$
        \STATE $\mbox{GM}_2 \leftarrow \bar{P}^a_i=\min\{P^a_i+B,0\}, \bar{R}^a_i = \max\{R^a_i, R^a_i+B+P^a_i\}$
        \STATE $LB_i \leftarrow \max_{j\in\{1,2\}}\mbox{ORIGAMI-BS}(GM_j)$
        \ENDFOR
        \STATE $globalLB \leftarrow \max_{i\in [n]}LB_i$\label{ln:GM_end}
        \STATE Sort $\mathcal P_i$ in decreasing $LB_i$.
        \FOR{sorted $\mathcal P_i$}
        \STATE Overuse Modifications (OM): $\forall j\neq i, \bar{R}^a_j = \max\{0, R^a_j-B\}, \bar{P}^a_j = \min\{P^a_j, P^a_j-B+R^a_j\}$
        \STATE $\mbox{OM}_1 \leftarrow \bar{R}^a_i=R^a_i+B.$
        \STATE $\mbox{OM}_2 \leftarrow \bar{P}^a_i=\min\{P^a_i+B,0\}, \bar{R}^a_i = \max\{R^a_i, R^a_i+B+P^a_i\}$ 
        \STATE $UB_i \leftarrow \min_{j\in\{1,2\}}\mbox{ORIGAMI-BS}(\mbox{OM}_j)$
        \IF{$UB_i \leq globalLB$}
        \STATE Prune $\mathcal P_i$
        \ELSE
        \STATE run MILP of $\mathcal{P}_i$ with additional constraint $c_i\leq \frac{globalLB-P^d_i}{R^d_i - P^d_i}$ 
        \ENDIF
        \ENDFOR
        \RETURN Best solution among $globalLB$ and all $\mathcal P_i$'s.
    \end{algorithmic} 
\end{algorithm}

We end this subsection by remarking that atomic payoff manipulation arises in many real-world applications. For example, it is infeasible for the wildlife ranger to charge the poacher a fine of \$$100/3$. In those cases, our proposed MILP formulation could be directly applied.

\subsection{PTAS for Limited Budget and Uniform Costs}
\label{subsect:l1ptas}
We show that for a special but practical class of problems, there exist a PTAS. In many applications, the defender has only a limited budget 
$B \leq \min_{j\in T} \{\left|P^a_j\right|, R^a_j\}$. Additionally, the weights on $\epsilon$ and $\delta$ are the same. W.l.o.g., we assume $\mu_j = \theta_j = 1$. We first show a structural theorem below and then discuss its algorithmic implication.

\begin{restatable}{thm}{twotarget}
\label{thm:l1-2target}
    When budget 
    $B \leq \min_{j\in T} \{\left|P^a_j\right|, R^a_j\}$
    and $\mu_j = \theta_j = 1, \forall j\in T$, there exists an optimal solution which manipulates the attack target and at most one other target. 
\end{restatable}

\begin{proof}[Proof sketch]
Since $B$ is limited, either $R_t^a$ or $P_t^a$ is unchanged according to Condition 3 of Property \ref{property:l1opt}. Assume all manipulations happen on attacker's reward. If some three targets get manipulated, we can simultaneously increase $\epsilon_{t}$ for attack target $t$ and decrease $\epsilon_j$ for $j\neq t$ such that $j$'s utility increases to be the same as target $t$, until some $\epsilon_j$ becomes $0$. After such change, the defender's utility does not decrease, and the number of targets manipulated decreases.
Other cases also hold due to symmetry.
\end{proof}

The theorem above is tight, i.e.\ we show in the full version an instance where two targets are manipulated. When 
$B \leq \min_{j\in T} \{\left|P^a_j\right|, R^a_j\}$
and $\mu_j = \theta_j = 1, \forall j\in T$, Theorem~\ref{thm:l1-2target} naturally suggests a PTAS -- we can use linear search for manipulations on all pairs of targets as shown in Alg.~\ref{algo:l1fptas}, where $e_i$ is a unit vector with a single one at position $i$. Theorem~\ref{thm:fptas-accuracy} shows the approximation guarantee, with a proof similar to Theorem~\ref{thm:l1milpapprox}, which is included in the full version.

\begin{algorithm}[t]
\scriptsize
\caption{PTAS for a special case in $L^1$}                       
    \begin{algorithmic}[1] \label{algo:l1fptas}
        \REQUIRE Payoffs $\{R^d, P^d, R^a, P^a\}$, budget $B$, tolerance $\eta$.
        \STATE Initialize $M \leftarrow -\infty$ 
        \FOR{all ordered pairs of targets $(i,j)$}
            \FOR{$s=0, 1, \ldots, \floor{B/\eta}$}
            \STATE $M \leftarrow \max\{M,\mbox{ORIGAMI-BS}(R^d, P^d, R^a+se_i-(B-s)e_j, R^d)\}$
            \STATE $M \leftarrow \max\{M,\mbox{ORIGAMI-BS}(R^d, P^d, R^a+se_i, R^d-(B-s)e_j)\}$
            \STATE $M \leftarrow \max\{M,\mbox{ORIGAMI-BS}(R^d, P^d, R^a, R^d+se_i-(B-s)e_j)\}$
            \STATE $M \leftarrow \max\{M,\mbox{ORIGAMI-BS}(R^d, P^d, R^a-(B-s)e_j, R^d+se_i)\}$
            \ENDFOR
        \ENDFOR
        \RETURN $M$
    \end{algorithmic}  
\end{algorithm}
\normalsize
\begin{restatable}{thm}{fptas}
\label{thm:fptas-accuracy}
    Alg.~\ref{algo:l1fptas} returns an additive $\max_{i\in[n]}\frac{2\eta(R^d_i-P^d_i)}{R^a_i}$ approximate solution.
\end{restatable}

\section{Optimizing Payoff with Budget Constraint in Other Forms}
In this section, we explore budget constraints in other forms and show polynomial time algorithms correspondingly. 

\subsection{Weighted $L^\infty$-norm Form}\label{subsec:linfty}
Consider the case where the defender can make changes to $R^a$ and $P^a$ for every target up to the extent specified by $B^r_i$ and $B_i^p$ respectively. Following previous notations, this requirement can be represented by a budget constraint in weighted $L^\infty$-norm form, i.e., 
$\max_j\{ |\epsilon_j|/B^r_i,|\delta_j|/B_i^p\}\leq 1$.
Equivalently, the defender can choose $R^a$ and $P^a$ from a given range. A real-world setting for this problem is when a higher level of authority specifies a range of penalty for activities incurring pollution and allow the local agencies to determine the concrete level of penalty for different activities.


We observe that Condition 2 of Property 1 still holds in this setting. Therefore, such problem can be solved by simply solving $n$ subproblems. In the $i^{th}$ subproblem which assumes $i$ is the attack target, we may set reward and penalty of $i$ to be the upper bound in the given range and choose the lower bound for other targets. With our improved ORIGAMI-BS algorithm, this problem can be solved in $O(n^2 \log n)$ time.
\begin{restatable}{thm}{linfty}
\label{thm:linfty}
    With budget constraint in weighted $L^\infty$-norm, solving for defender's optimal strategy reduces to solving for defender's optimal coverage in fixed-payoff security games.
\end{restatable}

\subsection{$L^0$-norm Form}\label{subsec:l0}
In some domains, the defender can make some of the targets special. For example, in wildlife protection, legislators can designate some areas as "core zones", where no human activity is allowed and much more severe punishment can be carried out. But the defender cannot set all the areas to be core zones. We model such restrictions as setting a limit on the Hamming distance between the original penalty vector and the manipulated penalty vector for the attacker, i.e., $\sum_j \mathds{1}(|\delta_j|>0)\leq B$ where $B$ is the budget. Following \cite{donoho2003optimally}, we refer to it as a $L^0$-norm form budget constraint for simplicity even though it is not technically a norm. That is, the defender needs to pay a unit cost to manipulate $P_i^a$ on target $i$ but the magnitude of change can be arbitrary. The defender needs to choose which targets to make changes. We do not consider the case where the defender can arbitrarily modify the attacker's reward $R^a$ as it is not practical and will lead to a trivial solution: the defender will place all coverage on one attack target $t=\arg\max_i R^d_i$ and set $R^a_t = \infty$. 

We assume the defender has a budget which allows her to change the penalty of $B$ targets. Similar to the $L^{\infty}$ case, we first observe that the defender will choose an extreme penalty value once he decides to change the penalty of a target.
\begin{property}\label{obs:L0Pvalue}
There exists an optimal solution where either $\bar P^a_j =-\infty$ for $B$ targets or $\bar P^a_j = -\infty$ for $(B-1)$ targets and $\bar P^a_t = 0$ for 1 target. If $\bar P^a_j = -\infty$, then $c_j = 0$.
\end{property}
\begin{proof}
When $t$ is the attack target, the defender would like to maximize $\bar P^a_t$ and minimize $\bar P^a_j$ for all $j\neq t$. If $\bar P^a_j = -\infty$ and $c_j > 0$, target $j$ will not be attacked as $U^a_j = -\infty$. In such case, target $j$ is effectively removed from the game.
\end{proof}

The defender's problem becomes non-trivial when the budget $B < T$, and we now provide a $O(n^3)$ algorithm (Alg.~\ref{algo:L0}) for solving this problem. We note that several intuitive greedy algorithms do not work, even in more restrictive game settings. A detailed comparison of our algorithm, several greedy algorithms, and a baseline MILP is provided in Section~\ref{sect:evalution}.

First, we sort the targets in decreasing attacker's reward $R^a_k$. Let $E_k = \frac{1}{R^a_k - P^a_k}$ and $\bar E_k = \frac{1}{\bar R^a_k - \bar P^a_k}$ for all $k \in T$. When $i$ is the attack target, by Property \ref{obs:L0Pvalue}, we have $\bar E_i \in \{1/R^a_i, E_i\}$ and $\bar E_j \in \{0, E_j\}$. Let $\Gamma_l = \{1,2,\dots,l\}$ for $l = 1,2,\dots,n$. 
We notice that one of the $\Gamma_l$'s, denoted as $\Gamma_{l^*}$, \emph{encapsulates} the attack set in the optimal solution to our problem. That is, in the optimal solution, each target in the attack set is in $\Gamma_{l^*}$; those targets not in the attack set either are outside $\Gamma_{l^*}$, or, if they are in $\Gamma_{l^*}$, have $\bar P^a = -\infty$. A proof can be found in the full version.
This allows us to formally define a subproblem $Q_{l, i}$: assume (i) the optimal attack set is
encapsulated by
$\Gamma_l$, (ii) the attack target is $i \in \Gamma_l$, and (iii) no target is covered with certainty, what is the defender's optimal strategy $(c, \bar E)$?
A subproblem may be infeasible. First, we show that $Q_{l,i}$ can be solved in $O(n)$ time. From Equation~\ref{eqn:cj}, for subproblem $Q_{l,i}$, we have
\begin{equation}
\frac{c_i}{\bar E_i} = \frac{\sum_{k \in \Gamma_l\backslash\{i\}} (R^a_i - R^a_k)\bar E_k + r}{\bar E_i + \sum_{k \in \Gamma_l\backslash\{i\}} \bar E_k} \label{eqn:quotient}
\end{equation} 

Let $s = \min\{B, l-2\}$ if $\bar E_i = E_i$ and $s = \min\{B-1,l-2\}$ if $\bar E_i = \frac{1}{R_i}$. Then $Q_{l, i}$ reduces to finding $s$ out of the $(l-1)$ $\bar E_k$'s to set to 0, and set the rest $\bar E_k = E_k$, so as to maximize the above quotient.
As a result, $Q_{l,i}$ is closely connected to the problem of choosing subsets with maximum weighted average, which can be solved efficiently.

\begin{proposition}~\cite{Eppstein:1997:Algorithms:weighted-average}\label{lemma:weighted-average}
    Given a set $S$ where $|S| = n$, real numbers $\{v_k: k\in S\}$, positive weights $\{w_k: k \in S\}$, and an integer $r$. Among all subsets of $S$ of order $n-r$, a subset $T \subset S$ which maximizes $A(T) = \frac{\sum_{k\in T} v_k}{\sum_{k\in T} w_k}$ can be found in $O(n)$ time.
\end{proposition}

\begin{lemma}
    The subproblem $Q_{l,i}$ can be solved in $O(n)$ time.
\end{lemma}
\begin{proof}
    Consider Equation \ref{eqn:quotient}. We equate $v_k =  (R_i - R_k) E_k + \frac{m}{l-s-1}$ and $w_k = E_k + \frac{1}{l-s-1} \bar E_i$. Let $v = \{v_k: k \in \Gamma_l\backslash\{i\}\}$, $w = \{w_k: k \in \Gamma_l\backslash\{i\}\}$. By Property~\ref{obs:L0Pvalue}, we may assume $s$ targets will be removed. Finding a subset $T \subset \Gamma_l\backslash\{i\}$, $|T| = l - s - 1$, to maximize $A(T)$ is equivalent to our problem to maximize the quotient in Equation~\ref{eqn:quotient}.
\end{proof}
After we find the optimal choices for the $s$ targets, we need to verify 
on Line~\ref{algostep:valid} of Alg.~\ref{algo:L0} 
that the attack set is valid. Since $c_k = 0$ for $k \notin \Gamma_l$, we need $M \geq R^a_{l+1}$, where $M$ is the attacker's expected utility as defined in Equation~\ref{eqn:M}. We also need valid coverage probabilities $c_j$'s. These could have been violated by setting some $\bar P_i$'s to $-\infty$. 

\begin{algorithm}[t]
\scriptsize
    \caption{Algorithm for budget in $L^0$-norm form}                       
    \begin{algorithmic}[1]  \label{algo:L0}
        \REQUIRE Payoffs $\{R^d, P^d, R^a, P^a\}$, budget $B$
        \STATE Initialize $U^d(1..n,1..n) \leftarrow -\infty$.
        \FOR{attack set $\Gamma_l$}
        \FOR{attack target $i \in \Gamma_l$}
        \STATE $\{Value, \Gamma_l^{drop}, \Gamma_l^{keep}\} \leftarrow \text{Random}(\langle v, w\rangle, s$) \label{algostep:random}
        \STATE Set $\bar E_j \leftarrow 0$ for $j \in \Gamma_l^{drop}$, $\bar E_j \leftarrow E_j$ for $j \in \Gamma_l^{keep}$.
        \STATE If solution is valid, i.e. $M \geq R^a_{l+1}$ and $c_j \in [0,1]$, then update $U^d(l,i)$\label{algostep:valid}
        \STATE Repeat inner iteration with $s \leftarrow \min\{B-1, l-2\}$ and $\bar E_i \leftarrow 1/R^a_i$.
        \ENDFOR
        \ENDFOR
        \FOR{target $k$ with largest $s$ $P^a_k$} \label{algostep:certaintystart}
        \FOR{attack target $i$}
        \STATE Update $U^d(l,i)$ if $r$ big enough for $c_k = 1$
        \ENDFOR
        \ENDFOR \label{algostep:certaintyend}
        \RETURN $\max(U^d)$
    \end{algorithmic} 
\end{algorithm}

We are now ready to show the main result of this section.

\begin{restatable}{thm}{lzero}
\label{thm:l0}
    There is a $O(n^3)$ algorithm for finding the optimal defender strategy with budget constraint in $L^0$-norm.
\end{restatable}

\begin{proof}[Proof sketch]
    Consider Alg.~\ref{algo:L0}. Since $R^a$ is fixed, $\Gamma_l$'s cover all the attack sets that need to be checked. 
    There are $O(n^2)$ subproblems $Q_{l,i}$. For each $Q_{l,i}$, we run a randomized algorithm for the maximum weighted average problem with expected running time $O(n)$ (Line~\ref{algostep:random}). A deterministic $O(n)$ algorithm exists in~\cite{Eppstein:1997:Algorithms:weighted-average}.
    The subproblems $Q_{l,i}$ miss the solutions where some target $j$ is covered with certainty. In this case, $\Gamma_n$ is the only possible attack set, and the solution is found on Lines~\ref{algostep:certaintystart}-\ref{algostep:certaintyend}. A solution is feasible if by removing targets we can keep the sum of coverage probabilities below the defender's resources.
\end{proof}

\section{Experimental Results}
\label{sect:evalution}
\setlength{\textfloatsep}{\textfloatsepsave}
\subsection{Simulation Results for $L^1$ Budget Problem}

We compare our branch-and-bound (BnB) algorithm (Alg.~\ref{algo:l1-pruning}) with three baseline algorithms -- NonConv, multiple MILP, and single MILP.
NonConv refers to solving $n$ non-convex optimization problems as shown in Equations~\ref{program:l1cont}-\ref{constr:nonnegative} using IPOPT~\cite{wachter:2006:MP:implementation} solver with default parameter setting, which converges to local optima with no global optimality guarantee. Multiple MILPs, as specified by Equations \ref{eqn:l1milpbegin}-\ref{eqn:APi-linearize2}, and the single MILP formulation, in the full version, are equivalent and have an approximation guarantee specified in Thm.~\ref{thm:l1milpapprox}. The original payoff structures are randomly generated integers between $1$ and $2n$ with penalties obtained by negation (recall $n$ is the number of targets). Budget and weights of the manipulations are randomly generated integers between $1$ and $4n$. 


We set $\rho_0 = \min_{i\in T}\frac{R^a_i}{4(R^d_i-P^d_i)}$ which gives an additive 
$\frac{1}{2}$-approximate solution. Gurobi is used for solving MILPs, which is terminated when either time limit (15 min) or optimality gap ($1\%$) is achieved. For each problem size, we run $60$ experiments on a PC with Intel Core i7 processor. The solution quality of a particular algorithm is measured by the multiplicative gap between that algorithm and BnB, i.e. $\frac{Z_\text{A}-Z_\text{BnB}}{Z_\text{BnB}}$ where $Z_\text{A}$ is best solution value by algorithm A. Thus a positive (negative) gap indicates better (worse) solution value than BnB. We report mean and standard deviation of the mean of runtime and solution quality in Fig.~\ref{fig:runtime}. Small instances refer to problem sizes from 5 to 25. Large instances refer to problem sizes from 50 to 250.

\begin{figure}[t]
    \subfloat[Runtime, small instances]{\includegraphics[clip, trim=1.15in 3.1in 1.2in 3.3in, width=1.4in]{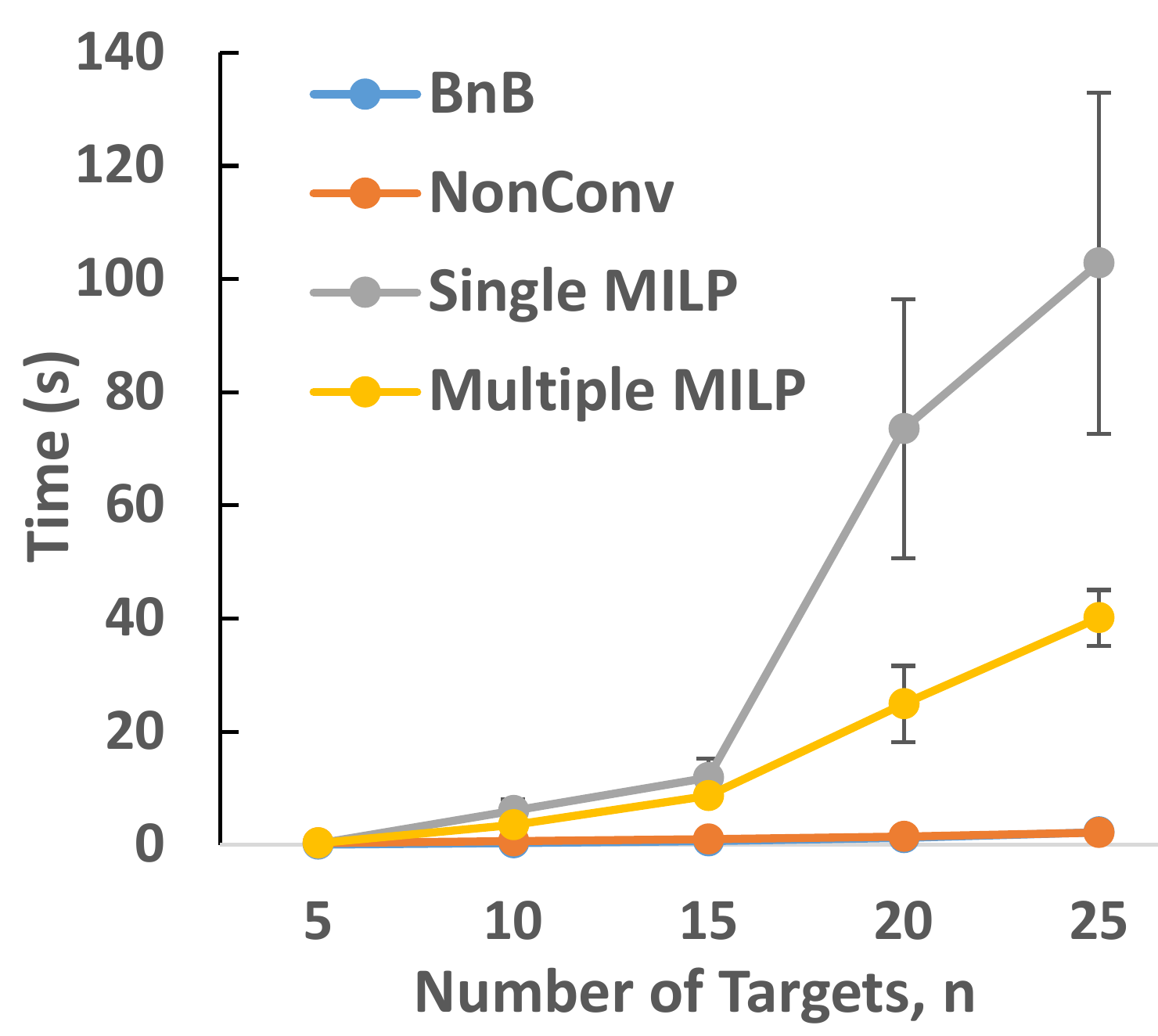}\qquad%
        \label{fig:time_small}}
    \subfloat[Gap, small instances]{\includegraphics[clip, trim=1.15in 3.05in 0.85in 3.3in, width=1.4in]{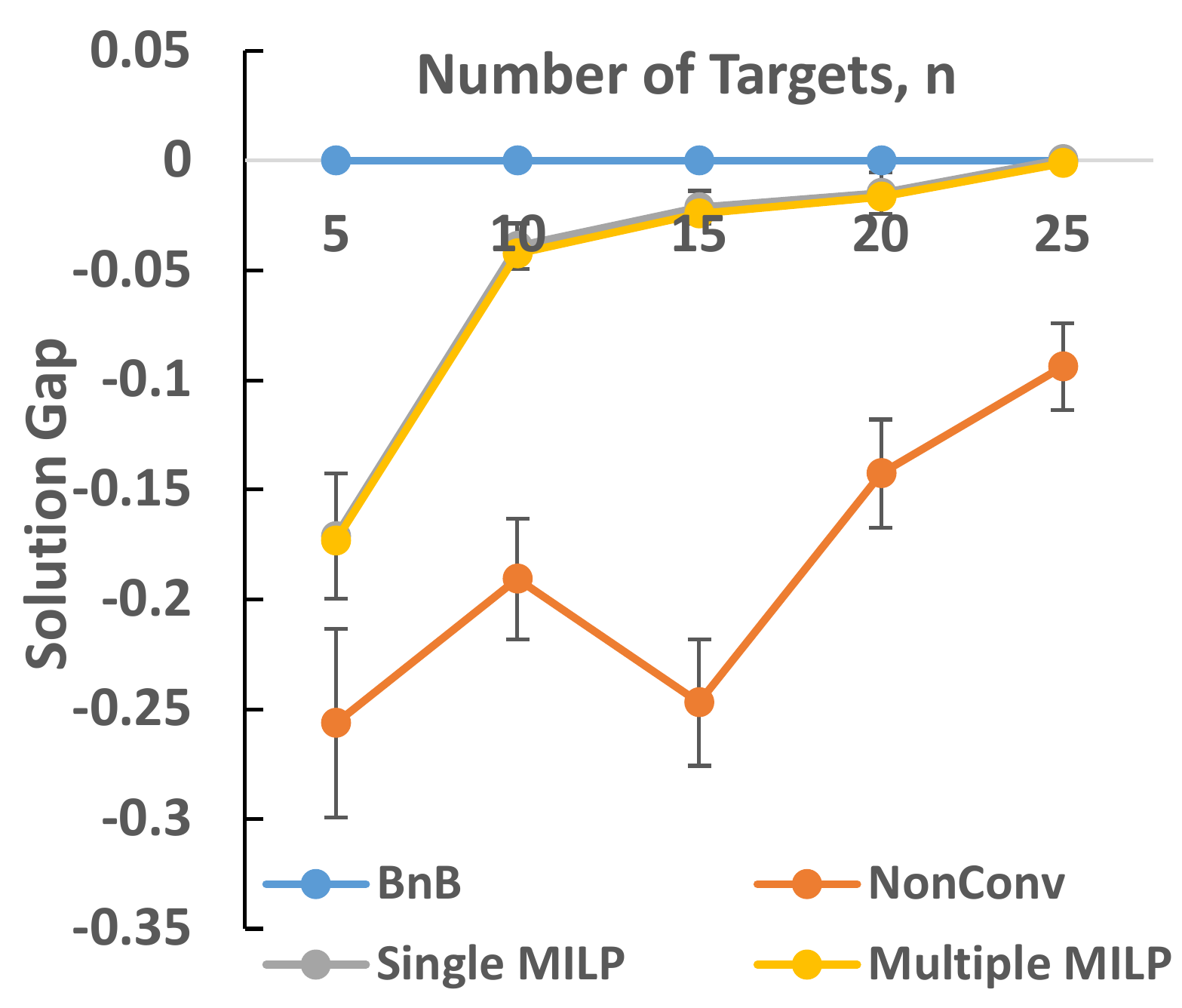}%
        \label{fig:gap_small}}\\
   \subfloat[Runtime, large instances]{\includegraphics[clip, trim=1.15in 3.3in 1.2in 3.5in, width=1.4in]{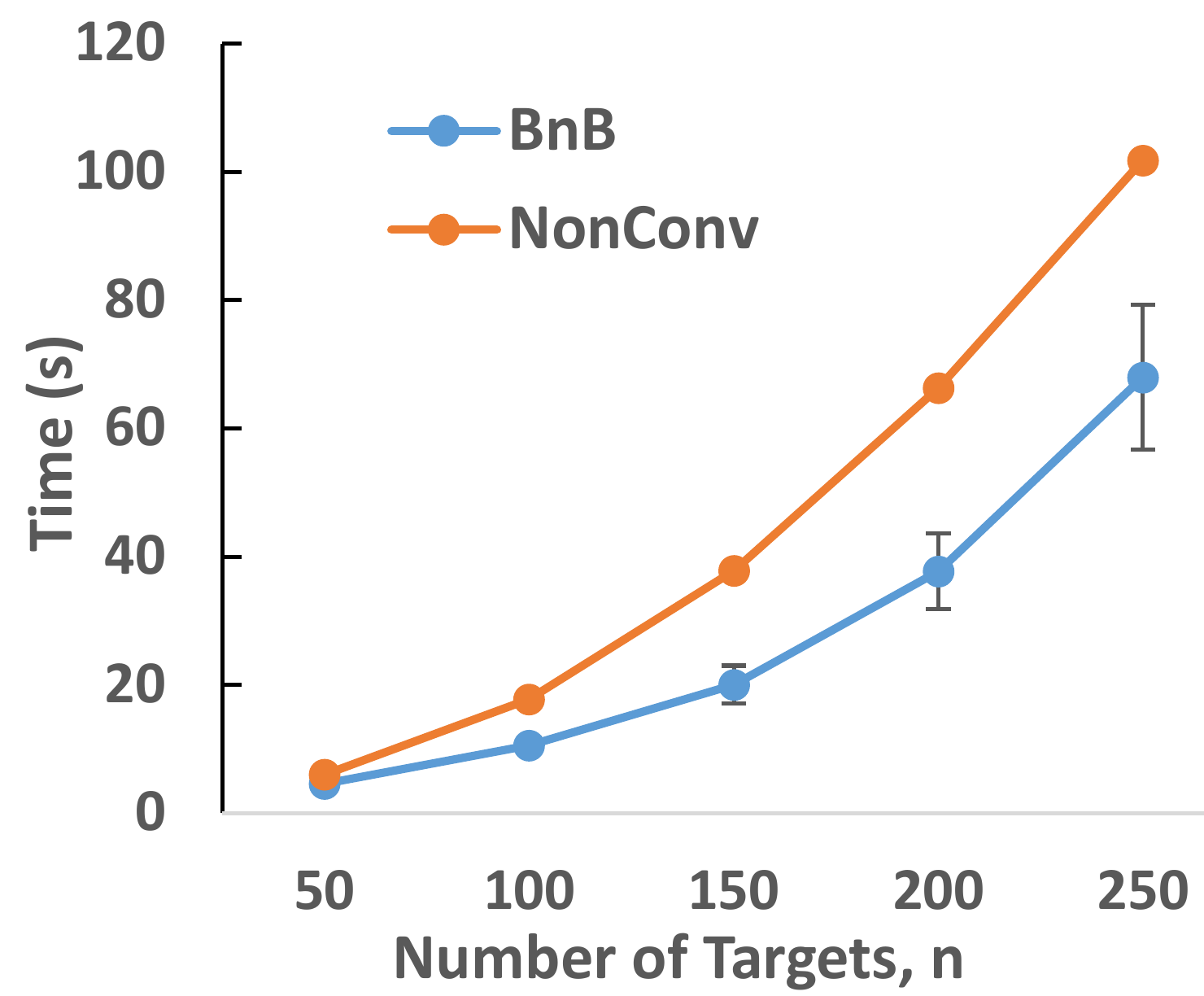}%
        \label{fig:time_big}}\qquad
   \subfloat[Gap, large instances]{\includegraphics[clip, trim=1.5in 3.3in 1.2in 3.3in, width=1.4in]{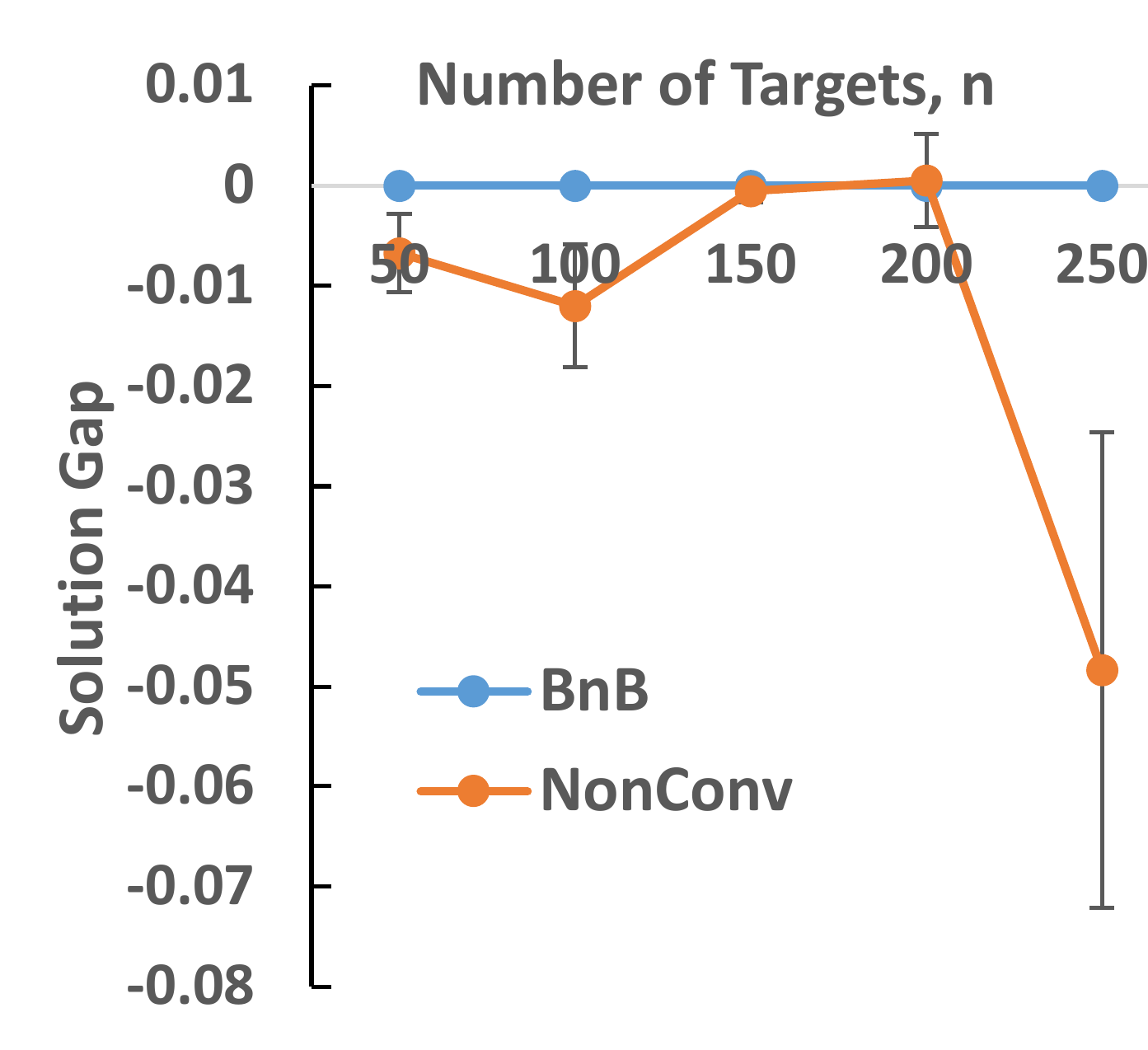}%
        \label{fig:gap_big}}
    \caption{Runtime and solution quality for $L^1$ case with standard deviation of the mean shown as vertical line}\label{fig:runtime}
\end{figure}
For problems of small size (Fig.~\ref{fig:time_small} and \ref{fig:gap_small}), BnB finds better solutions in nearly the same time as NonConv, faster than the other two. Since budget size can easily be indivisible by $\rho_0$ which is the atomic change we can make, greedy manipulation cannot be achieved by MILPs when such indivisibility happens. On the other hand, BnB first computes a global lower bound using such greedy manipulations, thus creating a gap between BnB and the other two MILP-based algorithms. Indeed the multiplicative gap between the greedy solution and the optimal solution is reported as $0.39\%$ with a variance of $0.14\%$. For problems of large size (Fig.~\ref{fig:time_big} and \ref{fig:gap_big}), we only compare BnB and NonConv as the other two algorithms timed out in solving MILP. BnB runs faster than NonConv. It returns better solutions for three problem sizes and nearly the same solution for the other two cases. The MILP-based solution including BnB also has a larger standard deviation in runtime than NonConv. 

\subsection{Simulation Results for $L^0$ Budget Problem}
We compare the performance of our $O(n^3)$ algorithm with a baseline MILP and two greedy algorithms. 
Greedy1 removes a target that can lead to most solution quality increase at a time. Greedy2 starts from the target with highest $|P^d|$ and determines whether to remove it by checking the solution quality before and after removal. Details of these algorithms can be found in the full version.

Initial payoffs are generated in the same way as in the previous subsection. In Fig.~\ref{L0_compare_r1}, we assume the defender has $r = 1$ resource and budget $B = n/2$, the worst case for the $O(n^3)$ algorithm. The runtime of MILP starts to explode with more than 100 targets, while the $O(n^3)$ algorithm solves the problem rather efficiently. We also note that MILP exhibits high variance in runtime. The variances of other algorithms, including the $O(n^3)$ algorithm, are relatively trivial and thus not plotted. We then test the algorithms with multiple defender resources, as shown in Fig.~\ref{L0_compare_rn10}. With $n$ targets, we assume the defender has $r = n/10$ units of resources and a budget $B = n/2$. Most MILP instances reach the time limit of 5 minutes when $n \geq 100$. Yet the $O(n^3)$ algorithm's runtime is almost the same as the single resource case. 

Our $O(n^3)$ algorithm and MILP are guaranteed to provide the optimal solution. In contrast, the greedy algorithms exhibit fast runtime but provide no solution guarantee. We measure the solution quality in Fig.~\ref{L0_greedy_r1} and~\ref{L0_greedy_rn10} using $\frac{U^d_{\text{greedy}} - p}{U^d_{\text{opt}} - p}$ where $p = \min_j P^d_j$. Greedy1, which runs slightly slower than Greedy2, achieves higher solution quality but both greedy algorithms can lead to a significant loss. In fact, extreme examples exist, as shown in the full version.

\begin{figure}[!t]\label{fig:L0_compare}
\subfloat[Resource $r = 1$]
{\includegraphics[clip, trim=1.1in 3.3in 1.2in 3.6in, width=1.4in]{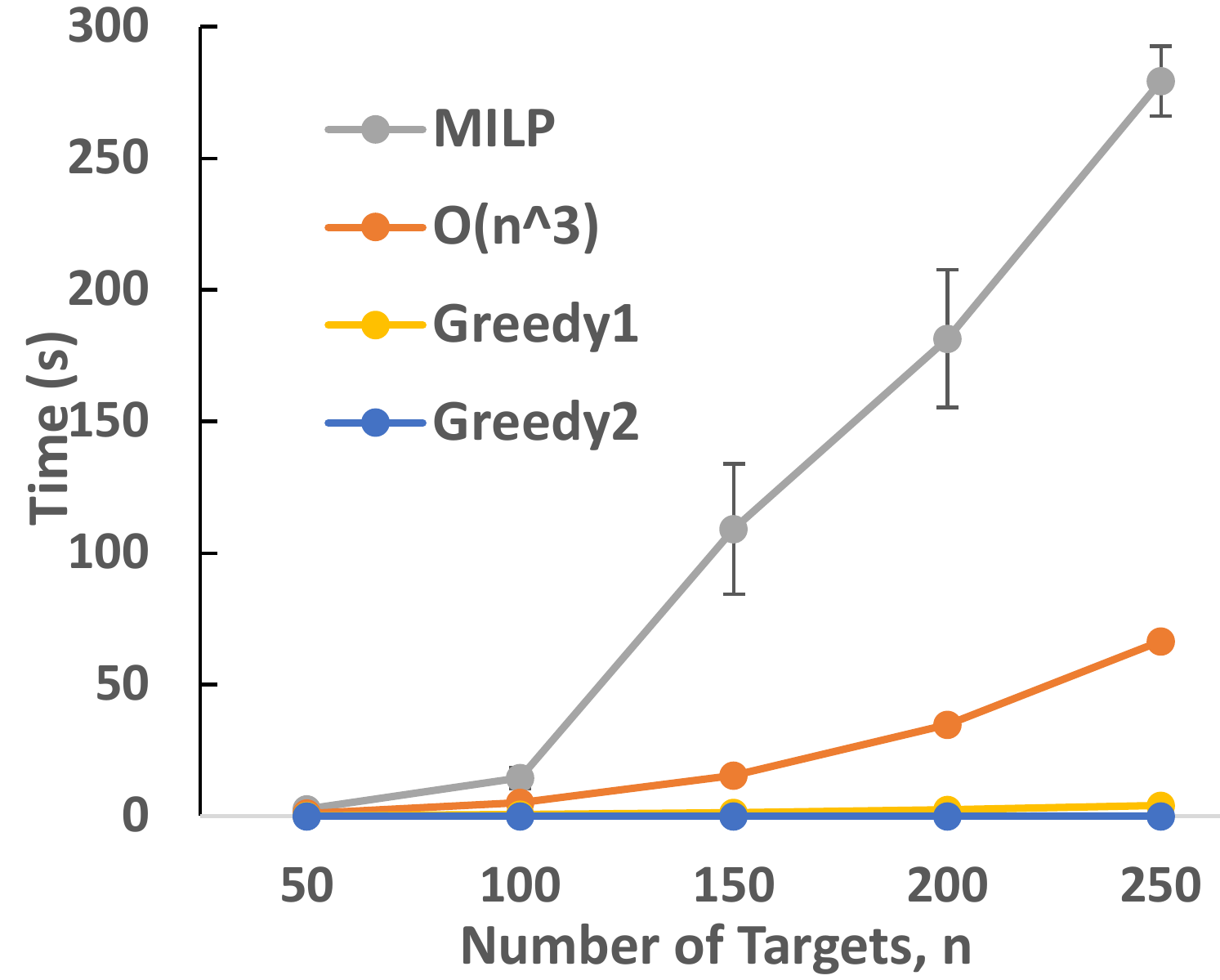}%
        \label{L0_compare_r1}} 
        \qquad
        \subfloat[Resource $r = n/10$]
{\includegraphics[clip, trim=1.1in 3.3in 1.2in 3.6in, width=1.4in]{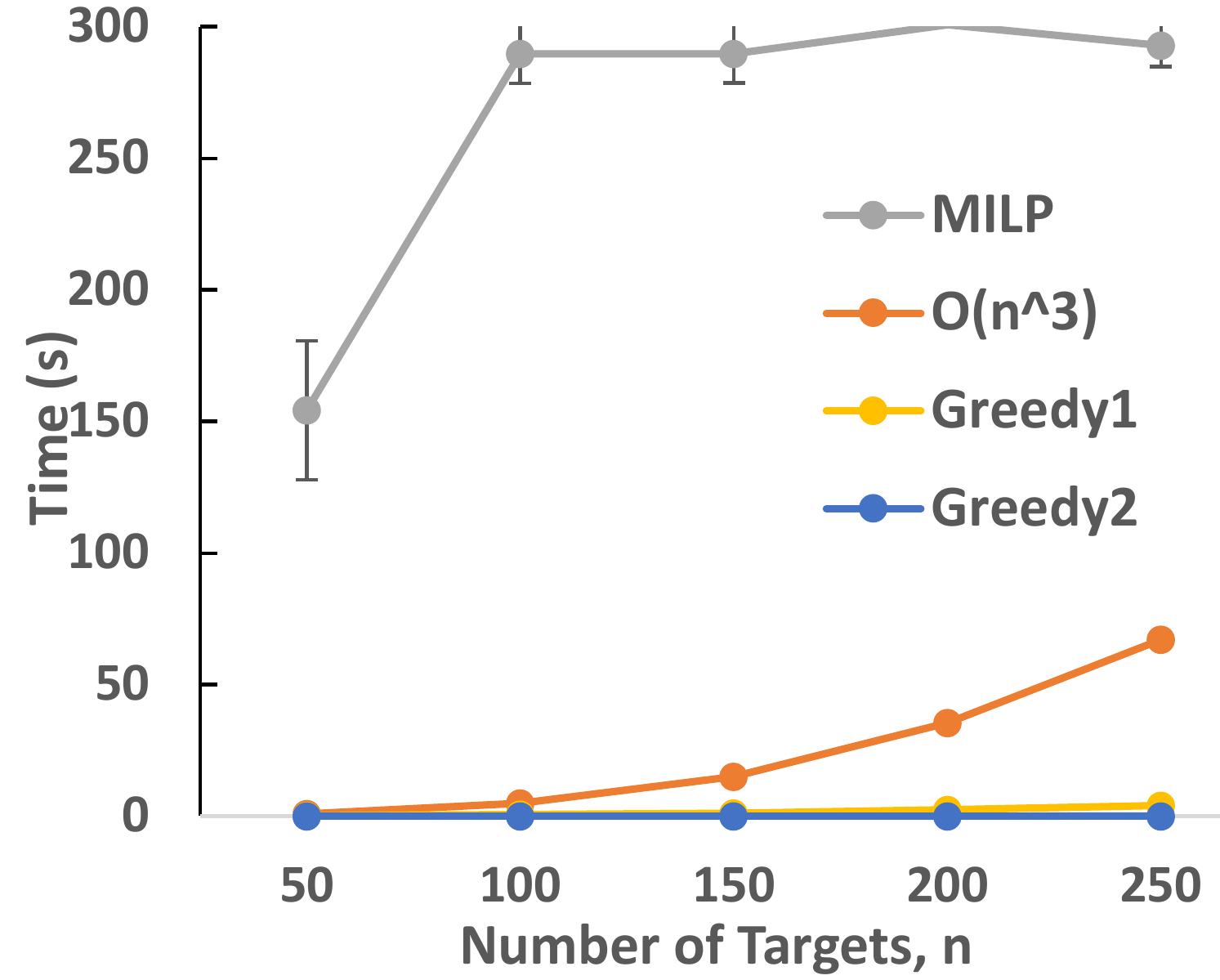}%
        \label{L0_compare_rn10}} \\ 
\subfloat[Resource $r = 1$]
{\includegraphics[clip, trim=1.1in 3.05in 1.1in 3.2in, width=1.4in]{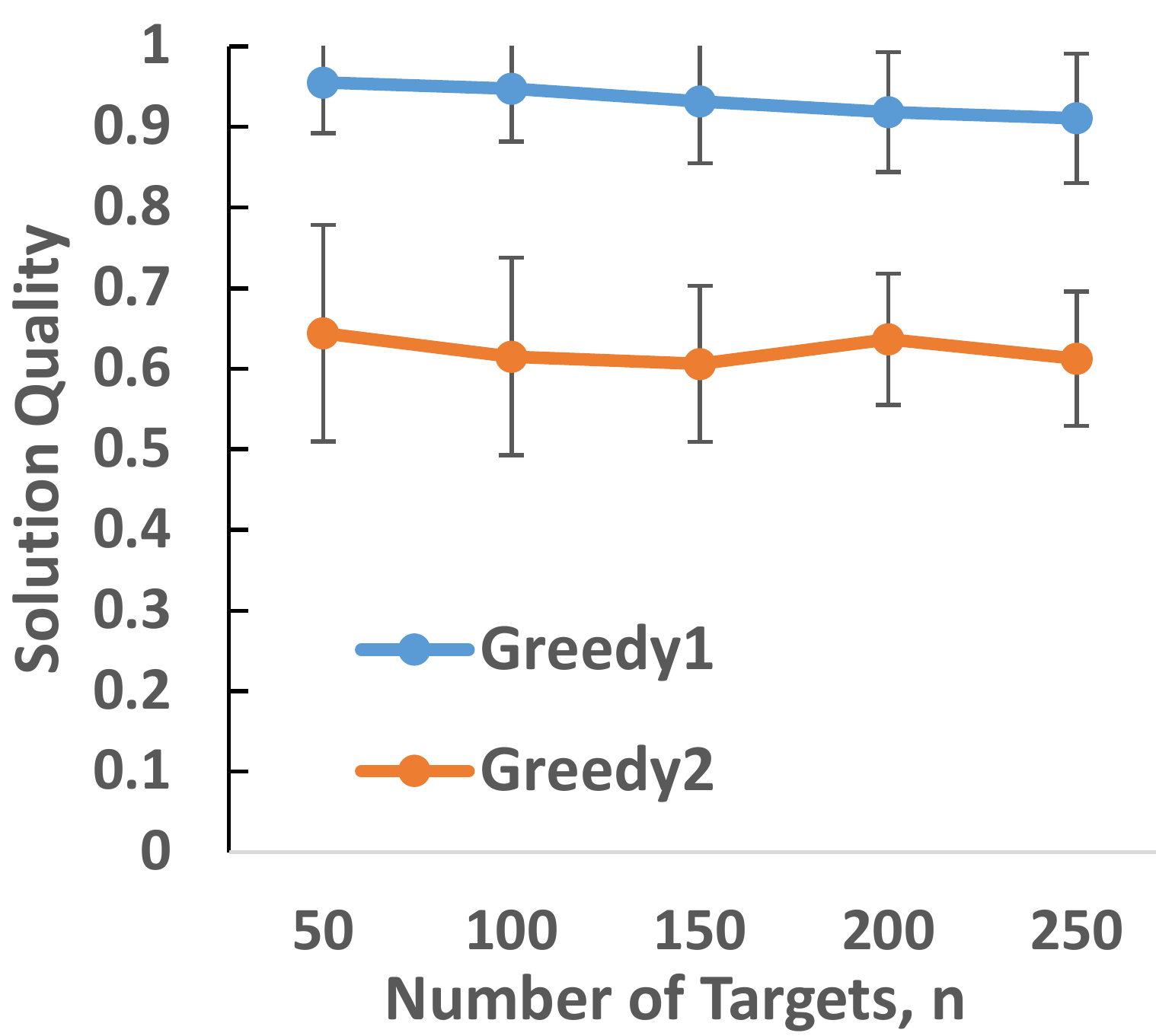}
        \label{L0_greedy_r1}} 
        \qquad
        \subfloat[Resource $r = n/10$]
{\includegraphics[clip, trim=1.1in 3.05in 1.1in 3.2in, width=1.4in]{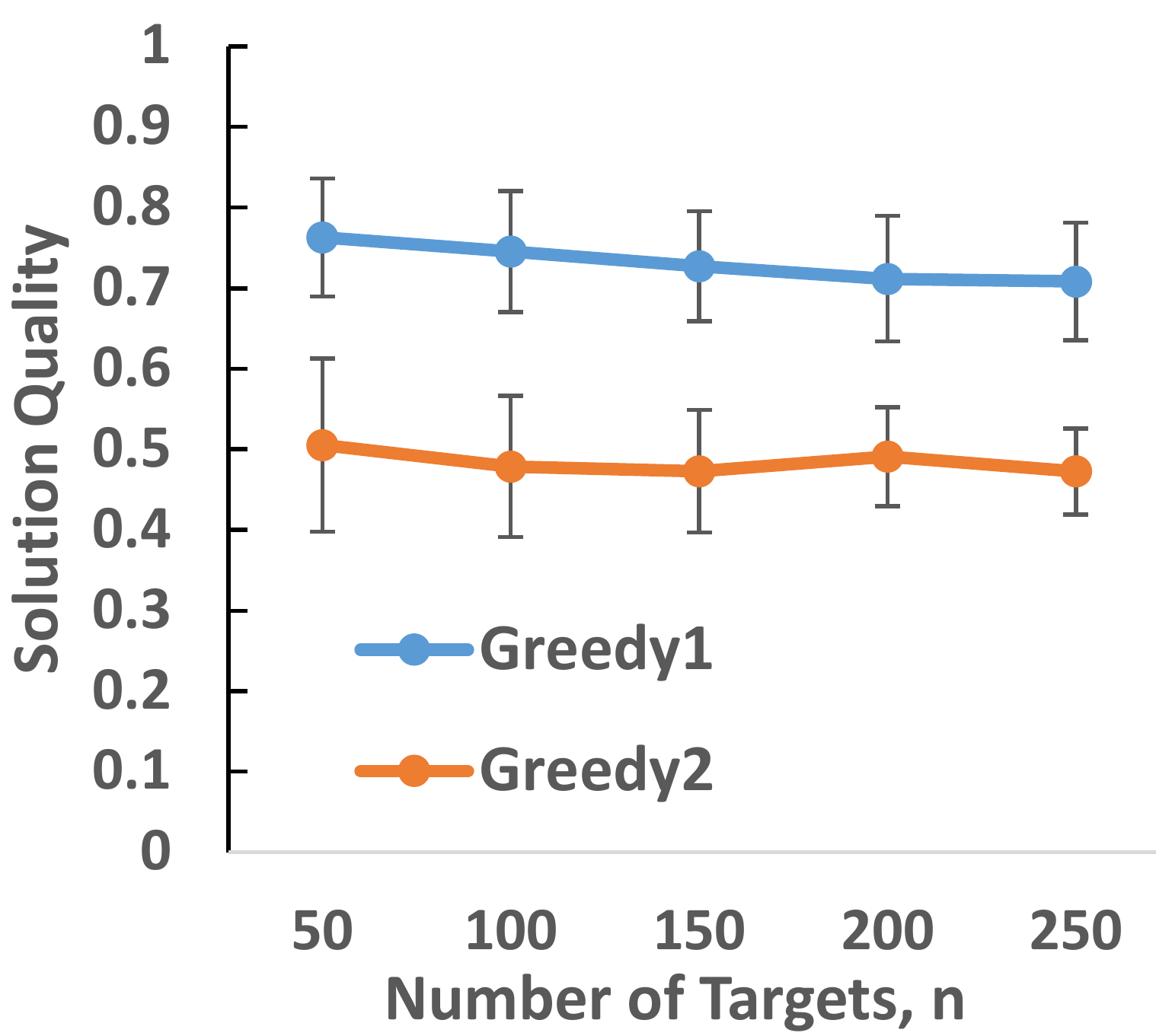}%
        \label{L0_greedy_rn10}} 
    \caption{
    Runtime and solution quality for $L^0$ case averaged over 22 trials. MILP has a time limit of 300 seconds. The error bars are standard deviations of the mean.
    }
\end{figure}

\section*{Acknowledgments}
The research is initiated with the support of the CAIS summer scholar program.

\small
\bibliographystyle{named}
\bibliography{ijcai18}
\normalsize

\clearpage
\appendix

\section{Omitted Algorithms}
\subsection{ORIGAMI with Binary Search} \label{sect:ORIGAMI-BS}
\begin{algorithm}[H]
	\caption{ORIGAMI-BS}                 \scriptsize
	\begin{algorithmic}[1] \label{algo:origami-binary-search}
		\REQUIRE Payoffs $\sigma = \{R^d, P^d, R^a, P^a\}$
		\STATE Initialize $start = 0, end=n, c=\mathbb{0}, lastTarget = n-1$.
		\WHILE{$end - start > 1$}
		\STATE Let $mid = \floor{\frac{end+start}{2}}$, attack set $\Gamma \leftarrow \{1,\ldots, mid\}$.
		\STATE Calculate $M$ and $c$ using Equations \ref{eqn:M} - \ref{eqn:cj}
		\IF{$c(k) < 0$ for some $k\in \Gamma$}
		\STATE $end\leftarrow mid$
		\ELSIF{$M < R^a_{mid+1}$}
		\STATE $start \leftarrow mid$
		\ELSE
		\STATE $lastTarget \leftarrow \max\{k\geq mid: R^a_k = M\}$
		\STATE break
		\ENDIF
		\ENDWHILE
		\STATE Let $\Gamma = \{1,\ldots, lastTarget\}$. Let $A = \{k\in\Gamma: c_k >= 1\}.$
		\IF{$A\neq \emptyset$}
		\STATE $covBound = \max\limits_{k\in A}R^a_k$. For all $k\in\Gamma$, set $c_k \leftarrow \frac{R_i-covBound}{R_i-P_i}$
		\ENDIF
		\FOR{$k\notin\Gamma$} 
		\STATE $c_k\leftarrow 0$
		\ENDFOR 
		\FOR{$k\in [n]$} 
		\STATE $defEU(k) \leftarrow c(k)$.
		\ENDFOR
		\RETURN $\max\limits_{k\in\Gamma}defEU(k)$
	\end{algorithmic} 
    \normalsize
\end{algorithm}

\subsection{A single MILP for $L^1$ Budget Problem}
\label{appendix:single-origami-milp}
\begin{align}
\min \qquad & d\\
s.t. \qquad &\epsilon_j = \rho_0\sum_k 2^k y^k_j-R^a_j \label{eqn:APi-binary1-appendix}\\
&\delta_j=-\rho_0\sum_k 2^k z^k_j-P^a_j \label{eqn:APi-binary3-appendix}\\
&-\epsilon^+_j \leq \epsilon_j \leq \epsilon^+_j\\ 
&-\delta^+_j \leq \delta_j \leq \delta^+_j\\
&\sum_{j=1}^n (\mu_j \epsilon^+_j+\sigma_j \delta^+_j) \leq B\\
&R_j^a+\epsilon_j \geq 0\\
&P_j^a+\delta_j \leq 0\\
&0\leq \alpha^k_j\leq y^k_j,\quad c_j -(1-y^k_j)\leq \alpha^k_j\leq c_j \label{eqn:APi-linearize1-appendix}\\
&0\leq \beta^k_j\leq z^k_j,\quad c_j-(1-z^k_j)\leq \beta^k_j\leq c_j \label{eqn:APi-linearize2-appendix}\\
&v_j = \rho_0\sum_k 2^k(y^k_j - \alpha^k_j - \beta^k_j )\\ 
&v_j \leq U \label{constr:upper_attEU-appendix} \\
&v_j \geq U-(1-\gamma_j)Z \label{constr:lower-attEU-appendix}\\
&d \leq R^d_jc_j + P^d_j(1-c_j) + (1-\gamma_j)Z\label{constr:att-defEU-appendix}\\
&\sum_j\gamma_j = 1\\
&y^k_j, z^k_j,\gamma_j \in \{0,1\}\\
&\sum_j c_j \leq 1, \quad c_j\geq 0
\end{align}
All constraints involving sub/super-script $j, k$ without a summation apply to all proper range of summation indices.

We first introduce non-negative integer variables $\epsilon^+_j$, $\delta^+_j$ and constraints $\epsilon^+_j\geq \epsilon_j$, $\epsilon^+_j\geq -\epsilon_j$, $\delta^+_j\geq \delta_j$, $\delta^+_j\geq -\delta_j$ to replace the absolute value change. We use binary representation for $\bar R^a$ and $\bar P^a$ in constraints~\ref{eqn:APi-binary1-appendix}-\ref{eqn:APi-binary3-appendix}. Specifically, we have $\bar R_j^a = \sum_k 2^k y^k_j, P_j^a=-\sum_k 2^k z^k_j$ where $y^k_j, z^k_j$ are 0-1 variables. Recall from Prop.~\ref{property:l1opt} we can assume without loss of generality that $\epsilon_t, \delta_t\geq 0$ and $\forall j\neq t, \epsilon_j,\delta_j\leq 0$. 

After the above reformulation, notice that we have bilinear term like $y^k_jc_j$ involved in the formulation. We introduce real-valued variables $\alpha^k_j, \beta^k_j$ and constraints~\ref{eqn:APi-linearize1-appendix}-\ref{eqn:APi-linearize2-appendix} to enforce $\alpha^k_j=y^k_jc_j$ and $\beta^k_j = z^k_j c_j$. We then introduce binary variables $\gamma_t$ to indicate whether target $t$ is in the attack set (the set of targets with highest attacker utility) and constraint~\ref{constr:upper_attEU-appendix}-\ref{constr:lower-attEU-appendix} to enforce that the attack target has the highest attacker's expected utility. Constraint~\ref{constr:att-defEU-appendix} enforces $d$ to be upper bounded by the defender's expected utility of the attack target. Therefore maximizing $d$ gives the defender's expected utility of the attack target.

\subsection{Baseline Algorithms for $L^0$ Budget Problem}\label{sect:l0appendix}
\begin{align}
\max \qquad & d\\
s.t. \qquad & d - U_\Theta(t, C) \leq (1 - a_t) Z,\qquad \forall t\\
& -w_t Z \leq k - U_\Psi(t, C) \leq (1 - a_t) Z, \qquad \forall t\\
& U_\Theta(t, C) = c_t R^d_t + (1-c_t) P^d_t, \qquad \forall t\\
& U_\Psi(t,C) \geq (1-c_t)R^a_t + c_t P^a_t - b_t Z, \qquad \forall t\label{eqn:l0milpAttEU1}\\
& U_\Psi(t,C) \leq (1-c_t)R^a_t + c_t P^a_t + b_t Z, \qquad \forall t\\
& U_\Psi(t,C) \geq (1-c_t)R^a_t + (1-b_t) P^a_t,\qquad \forall t\\
& U_\Psi(t,C) \leq (1-c_t)R^a_t + (b_t - 1) P^a_t,\qquad \forall t\label{eqn:l0milpAttEU4}\\
& c_t \in [0,1], \qquad \sum_t c_t \leq 1\\
& b_t \leq a_t, \qquad \forall t\\
& w_t + a_t \leq 1, \qquad \forall t\label{eqn:l0milpremove}\\
& \sum_t w_t + \sum_t b_t \leq B\\
& \sum_t a_t = 1\\
& a_t, w_t, b_t \in \{0,1\}
\end{align}
In this MILP, $a_t$ indicates whether target $t$ is attacked; $b_t$ indicates whether the attacker's penalty on target $t$ is set to 0; $w_t$ indicates whether target $t$ is removed (attacker's penalty on target $t$ is set to $-\infty$); $c_t$ is the coverage probability.

Equations~\ref{eqn:l0milpAttEU1}-\ref{eqn:l0milpAttEU4} ensure that for non-attack targets and for attack targets whose attacker's penalty is not set to 0, $U_\Psi(t,C) = (1-c_t)R^a_t + c_t P^a_t $; for the attack target whose attacker's penalty is set to 0, $U_\Psi(t,C) = (1-c_t)R^a_t$. This formulation assumes $P^a_t \leq 0$. Equation~\ref{eqn:l0milpremove} says the defender can only remove a non-attack target.

\begin{algorithm}
\scriptsize
	\caption{L0-Greedy1}                       
	\begin{algorithmic}[1]
		\REQUIRE Payoffs $\{R^d, P^d, R^a, P^a\}$, budget $B$
		\STATE Initialize $defEURemove(1..n) = defEUZero(1..n) = -\infty$.
		\FOR{greedy step $i$ = 1 \TO B}
		\FOR{each remaining target $t$ with $P^a_t \neq 0$}
		\STATE defEURemove(t) = ORIGAMI(Payoffs with target $t$ removed)	
		\ENDFOR
		\IF{$P^a_t \neq 0$ for all targets $t$}
		\FOR{each remaining target $t$}
		\STATE defEUZero(t) = ORIGAMI(Payoffs with $P^a_t= 0$)
		\ENDFOR
		\ENDIF
		\IF{max(defEURemove) $\geq$ max(defEUZero)}
		\STATE Remove the target with max(defEURemove)
		\ELSE 
		\STATE Set $P^a_t = 0$ for the target $t$ with max(defEUZero)
		\ENDIF 
		\ENDFOR
		\RETURN max(max(defEURemove), max(defEUZero))
	\end{algorithmic}
\normalsize
\end{algorithm}

The baseline L0-Greedy1 is not an exact algorithm, which can give arbitrarily large error with the following example.

Suppose we have 3 targets, budget is 2. Consider the following payoff matrix on the left. The optimal solution is to eliminate targets 1 and 2, extracting all $R^d_3 = 10$ as defender's utility. Yet L0-Greedy1 sets $P^d_2 = 0$ in the first iteration. In the second iteration, it is indifferent about removing targets 1 or 3, both giving a final defender's utility of 0.476. We can scale up all values with absolute value 10 to arbitrarily large, and the resulting defender's utility is still very small.

\begin{tabular}{|c|c|c|c|}
	\hline
	& $t_1$ & $t_2$& $t_3$\\
	\hline 
	$R^d$	& 1	&1	&10\\ 
	\hline 
	$P^d$	&-1	&-10& 	-10\\ 
	\hline 
	$R^a$	&1	&10&	1 \\ 
	\hline 
	$P^a$	&-10&	-10&	-10\\
	\hline
\end{tabular} 
\begin{tabular}{|c|c|c|c|}	
	\hline
	& $t_1$ & $t_2$& $t_3$\\
	\hline 
	$R^d$	& 10	&1	&1\\ 
	\hline 
	$P^d$	&-1.1	&-1 & 	-0.9\\ 
	\hline 
	$R^a$	&1	&1&	1 \\ 
	\hline 
	$P^a$	&-1&	-1&	-1\\
	\hline
\end{tabular}

The L0-Greedy2, as shown in Alg.~\ref{algo:l0greedy2}, is also not exact. Consider the above payoff matrix on the right. The optimal solution is to eliminate targets $t_2$ and $t_3$, extracting $U^d = R^d_1 = 10$. However, since ORIGAMI yields a uniform coverage, target $t_1$ is removed at the first iteration, thus the defender can achieve utility of 1 at maximum.

\begin{algorithm}
\scriptsize
	\caption{L0-Greedy2}                       \label{algo:l0greedy2}
	\begin{algorithmic}[1]
		\REQUIRE Payoffs $\{R^d, P^d, R^a, P^a\}$, budget $B$
        \WHILE{Budget allows}
        \STATE Solve the game with ORIGAMI
        \STATE If $j$ is in the attack set, remove it, $j \leftarrow j + 1$.
        \ENDWHILE
		\RETURN max(max(defEURemove), max(defEUZero))
	\end{algorithmic}
\normalsize
\end{algorithm}

\section{Full Proofs Omitted in Text}\label{append:proof}
\optproperty*
\begin{proof}
	\textbf{Condition 1} In an optimal solution $(c,\epsilon,\delta)$, by the principle of ORIGAMI, we know that $c_j = 0, \forall j \notin \tau(c,\epsilon,\delta)$. We build an optimal solution $(c, \epsilon^1,\delta^1)$ and show that it satisfies Condition 1. If $\epsilon_j > 0$ or $\delta_j \neq 0$, we can safely let $\epsilon_j^1 = 0$ and $\delta_j^1 = 0$ and target $j$ is still outside $\tau(c,\epsilon,\delta)$. If $\epsilon_j < 0$, we can increase $\epsilon_j$ to $\epsilon_j^1$, such that $j$ gets included in the attack set with $c_j = 0$. Hence, the solution $(c,\epsilon^1,\delta^1)$ is optimal and satisfies Condition 1. From now on, we assume such an optimal solution $(c, \epsilon^1, \delta^1)$ exists.
	
	\noindent\textbf{Condition 2} We only prove $\epsilon_t \geq 0$, others follow similarly. Suppose an optimal solution $(c,\epsilon^1,\delta^1)$ satisfying Condition 1 is such that $\epsilon_t^1 < 0$. Let $\delta^2 = \delta^1, \forall j\neq t, \epsilon^2_j = \epsilon^1_j$ and $\epsilon^2_t = -\epsilon^1_t$. Note that $(c, \epsilon^2, \delta^2)$ is a feasible solution with the same objective $c_i$ and hence optimal. Let $c^2$ be the coverage determined by ORIGAMI with payoff structure $(\epsilon^2, \delta^2)$, then once again $(c^2, \epsilon^2, \delta^2)$ is optimal. We apply the same reasoning as the previous paragraph, and get that some $(c^2,\epsilon^3,\delta^3)$ is an optimal solution satisfying condition 1 and $\epsilon^3_t \geq 0, \epsilon^3_j \leq 0$ for $j\neq t$.
	
	\noindent\textbf{Condition 3} We only prove the second part; the first part follows similarly. Now we have an optimal solution $(c^2,\epsilon^3,\delta^3)$ satisfying Conditions 1 and 2. Note that when some $c^2_j \leq 1/2$, having $\epsilon^3_j > 0$ and $\delta^3_j = 0$ yields the same $U^a_j$ as having $\epsilon^4_j = 0$ and $\delta^4_j = \epsilon^3_j (1-c^2_j)/c^2_j$. However, the former choice is more budget efficient. This holds unless the required change is excessive such that $R^a_j - \epsilon^3_j = 0$ and $\delta^3_j < 0$. Therefore, an optimal solution $(c^2, \epsilon^3,\delta^3)$ satisfying Conditions 1 and 2 can be easily modified to an optimal solution $(c^2, \epsilon^4, \delta^4)$ which also satisfies Condition 3.
%
\end{proof}

\begin{proof}[Proof of the non-convexity of Subproblem $\mathcal{P}_i$]\label{pf:l1_non_convex}
	Consider the constraints involving quadratic terms, we can rewrite the constraint in terms of variable $X = (c,\epsilon,\delta,1)$ where the last constant $1$ is used to generate linear terms. Using $X$, each quadratic constraint can be written in the form of $X^TQX\leq 0$ where all diagonals of $Q$ except the last one is zero. It is not hard to see that $Q$ is indefinite.
\end{proof}

\milpapprox*
\begin{proof}
The floor and ceiling notations in the sequel are about the ``integral grid'' defined by $\rho_0$. Let $(c^*, \epsilon^*, \delta^*)$ be an optimal solution to the subproblem $P_i$. We construct a feasible solution $(c', \epsilon', \delta')$ with $c'_i$ close to $c_i^*$. Let $U^{a*}_i$ and $U^{a'}_i$ be the attacker's expected utilities.

Let $\epsilon' = \lfloor \epsilon^* \rfloor$, $\delta' = \lfloor \delta^* \rfloor$, and $c' = c^*$ except $c_i' = c_i^* - \frac{2\rho_0}{D_i + \epsilon_i' - \delta_i'}$. For the attack target $i$, we have
\begin{equation}
\begin{split}
&U^{a*}_i - U^{a'}_i \\
&= (\epsilon_i^* - \epsilon_i') - D_i(c_i^* - c_i') - (c_i^* \epsilon_i^* - c_i'\epsilon_i') + (c_i^* \delta_i^* - c_i'\delta_i')\\
&= (1-c_i^*) (\epsilon_i^* - \epsilon_i') + c_i^*(\delta_i^* - \delta_i') - 2\rho_0\\
&\leq -\rho_0
\end{split}
\end{equation}
For non-attack targets $j$, we have
\begin{equation}
\begin{split}
&U^{a*}_j - U^{a'}_j \\
&= (\epsilon_j' - \epsilon_j^*) - D_i(c_j^* - c_j') + (c_j^* \epsilon_j^* - c_j'\epsilon_j') - (c_j^* \delta_j^* - c_j'\delta_j')\\
&= (1-c_j^*) (\epsilon_j' - \epsilon_j^*) - c_j^*(\delta_j^* - \delta_j')\\
&\geq -\rho_0
\end{split}
\end{equation}

Therefore, the solution $(c', \epsilon', \delta')$ is feasible, and we have $c_i' \geq c_i^* - \frac{2\rho_0}{R^a_i}$. By solving all atomic subproblems, we have an additive $\max_i (R^d_i - P^d_i) \frac{2\rho_0}{R^a_i}$-approximation to the original problem.

\end{proof}

\twotarget*
\begin{proof}
Since the budget is limited, for each target exactly one of its reward and penalty can be manipulated, i.e. $\forall i\in T$ exactly one of $\epsilon_i$ and $\delta_i$ is non-zero. Let $\Gamma$ be the set of targets whose reward or penalty is manipulated, let $t$ be the attack target. Below we only consider the case where reward has been manipulated. Others follow similarly due to symmetry. Let $i^* =\arg\min_{i\in \Gamma\backslash\{t\}}\epsilon_i(1-c_i)$. For each $i\in \Gamma\backslash\{t\}$, decrease $\epsilon_i$ by $\frac{(1-c_i)\epsilon_{i^*}}{1-c_{i^*}}$. Note by the definition of $i^*$ we have $\epsilon_i \geq \frac{(1-c_i)\epsilon_{i^*}}{1-c_{i^*}}$, i.e. we won't make some reward negative when decreasing decreasing $\epsilon$ as above. Now we increase the reward of attack target $t$ by $\frac{(1-c_t)\epsilon_{i^*}}{1-c_{i^*}}$. It is not hard to see after those manipulations, those targets that were in the attack set before are still in the attack set and that the objective value does not decrease. Below we show that it is feasible to increase the reward of attack target by $\frac{(1-c_t)\epsilon_{i^*}}{1-c_{i^*}}$. To do so we only need to show $\sum_{i\in \Gamma\backslash\{t\}} \geq 1-c_t$. Note by decreasing the manipulation on other targets, we obtain $\sum_{i\in \Gamma\backslash\{t\}}\frac{\epsilon_{i^*}(1-c_i)}{1-c_{i^*}}$ available budgets. Since we assume only reward is manipulated, it must be that $c_i \leq \frac12, \forall i\in T$. As a result $\sum_{i\in \Gamma\backslash\{t\}}(1-c_i)\geq \sum_{i\in \Gamma\backslash\{t\}}c_i\geq 1-c_t$.
\end{proof}

\fptas*
\begin{proof}
Suppose the optimal solution has $i$ as the attack target and $j$ is the only other target in the attack set that has been manipulated. Let the optimal solution be $c^*, \epsilon^*, \delta^*$. w.l.o.g. assume rewards are manipulated, i.e. $\epsilon^*_i > 0, \epsilon^*_j < 0$. Other cases follow similarly. Let $q$ be an integer such that $q\eta \leq \epsilon^*_i < (q+1)\eta$. Below we show the solution $c$ where $c_i = c^*_i - \frac{\eta}{D_i}$ and $\forall k\neq i, c_k = c^*_k$ is a feasible solution when $\epsilon_i = q\eta, \epsilon_j = q\eta-B$. The theorem then follows.

When $\epsilon_i = q\eta, \epsilon_j = q\eta-B$ with coverage probability exactly $c^*$, $attEU(i)$ decreases by $(1-c^*_i)(\epsilon^*_i - \epsilon_i) \leq \eta$ by our choice of $q$. To compensate this decrease to obtain a feasible solution, we decrease $c^*_i$ by $\frac{\eta}{D_i}$, which increases $attEU(i)$ by $\frac{\eta}{D_i}(D_i + \epsilon_i) > \eta$.  Also note that decreasing $R^a_j$ by $\epsilon_j$ maintains feasibility.
\end{proof}

\linfty*
\begin{proof}
	With target $i$ being attacked, we have $U^a_i \geq U^a_j$ for all $j \in T$. The defender maximizes $c_i$.
	\begin{equation} \label{eqn:linftyci}
	c_i = \max_{c, \bar P^a, \bar R^a} \min_{j \in T-\{i\}}  \frac{\bar R^a_i - \bar R^a_j + (\bar R^a_j - \bar P^a_j) c_j}{\bar R^a_i - \bar P^a_i}
	\end{equation}
	Let $(c, \hat P^a, \hat R^a)$ be an optimal solution. Let $\bar{P^a}$ be such that $\bar{P^a_i} = \min(0, P^a_{i}-B^p_i)$ and $\bar{P^a_j} = P^a_{i}+B^p_i$. Let $\bar{R^a}$ be such that $\bar{R^a_i} = R^a_{i}+B^r_i$ and $\bar{R^a_j} = \max(0, R^a_{i})$. Note that $(c, \bar{R^a}, \bar{P^a})$ is also an optimal solution. It is obvious that minimizing $|P^a_i|$, maximizing $|P^a_j|$, and minimizing $R^a_j$ should make the inequality still hold. To see that the defender should maximize $\bar R^a_i$, we can rewrite the RHS of Equation~\ref{eqn:linftyci} as
	\begin{equation}
	c_i \leq 1 + \frac{\bar P^a_i - \bar R^a_j (1 - c_j) - c_j \bar P^a_j}{\bar R^a_i - \bar P^a_i}, \qquad \forall j \in T
	\end{equation}
	If the numerator is nonnegative, the value of $R^a_i$ does not matter because $c_i$ as probability takes a maximal value of 1. Otherwise, we see that maximizing $\bar R^a_i$ maximizes the RHS. Therefore, the defender may always use $\bar{R^a}$ and $\bar{P^a}$ for the attacker's reward and penalty and solve the linear program for coverage probability $c$, as in fixed-payoff security games.
\end{proof}

\lzero*
\begin{proof}
	By the ORIGAMI algorithm, we need only check the attack sets $\Gamma_l = \{1,2,3,\dots,l\}$ for $l = 1,2,\dots,n$. Since $R^a$ is fixed, the attack sets that need to be checked in the optimal payoff structure (i.e. after manipulation) are contained in the attack sets $\Gamma_l$'s that we will check in the initial payoff structure. For example, suppose that in the final optimal solution, the defender removes targets 1,3,4, and the attack set is $\{2,5,6\}$ with target 7 left outside the attack set. Then, when we check the attack set $\{1,2,3,4,5,6\}$ in the initial payoff structure, we will find this solution.
	
When solving subproblem $Q_{l,i}$, if the choice violates $M \geq R^a_{l+1}$ or $c_j \geq 0$, we drop it on Line~\ref{algostep:valid} of Algorithm~\ref{algo:L0}. By ORIGAMI, the general optimal choice $\bar E^*$, which is discovered as an optimal choice for Equation \ref{eqn:quotient}, will also satisfy these conditions.
    
The subproblems $Q_{l,i}$ miss out solutions when some target $k$ is covered with certainty. ORIGAMI shows that in this case, $P^a_k \geq P^a_j$ for all $j \in \Gamma_l$. Furthermore, since $R^a_j \geq 0 \geq P^a_j$, the only possible attack set is $\Gamma_n = T$. Once such a target $k$ is fixed, the coverage $c_i$ on the attack target $i$, the objective value, is also fixed. We need only check whether the defender has enough resources to maintain the attack set after removing the most costly targets.
\end{proof}

\section{An instance where greedy modification for $L^1$ case is sub-optimal}
\label{appendix:counterexample}
Our goal is to find instances for the greedy manipulation to fail. We assume only reward is manipulated in the counterexample. Since defender's payoff can be arbitrarily set, we assume $t$ is the optimal attack target. By Thm.~\ref{thm:l1-2target} there is at most one more target manipulated in the optimal solution. We assume $j$ is that target. First note $1-c_t < 1-c_j$ as otherwise we won't push to $j$. We show the idea of shifting manipulation from $t$ to $j$ so that the solution value increases. Use $k$ to denote other targets in the attack set. For simplicity let's consider the case that in the greedy manipulation, all targets appear in the attack set, which is easy to construct. Denote $c_t, \epsilon_t$ to be solutions to the greedy manipulation. Let $\deps_i$ be the amount of $\epsilon_i$ to push to $j$. 

We can visualize the shifting process as follows: at the start we have greedy manipulation solution, where each target has a bin of the same height as their expected utility. Then we shift $\deps_t$ to $j$ which lowers the bin $t$ and $j$. Because $1-c_t < 1-c_j$, bin $j$ is lower than bin $t$. Then we decrease $c_j$ to lift bin $j$ and increase $c_k$ to lower bin $k$. Let $\dc_j$, $\dc_k$ be the absolute change. If $\dc_j > \sum_{k\in\Gamma\backslash\{t,j\}}\dc_k$, we have additional unused coverage which can be redistributed to lower all bins, which as a result increases $c_t$. Now quantify this process and see how to make the numbers work out. Below I will write $\Delta EU$ as absolute change in attacker's expected utility.

First note $\Delta EU(t) = \deps_t(1-c_t)$ and similarly for target $j$. To push bin $k$ down to bin $t$, the least increase in $c_k$ satisfies $\dc_kD_k = \Delta EU(t)\Rightarrow\dc_k = \frac{\deps_t(1-c_t)}{D_k}$. On the other hand we want to lift up bin $j$, the least decrease in $c_j$ satisfies $\deps_j(D_j-\deps_t) = \Delta EU(j) - \Delta EU(t)\Rightarrow \dc_j = \frac{(c_t-c_j)\deps_t}{D_t-\deps_t}$. We require 
\begin{align}
\dc_j \leq c_j\Rightarrow\deps_t \leq \frac{c_jD_j}{c_t}\label{constr:enough-cj}
\end{align}

Finally the constraint $\sum_{k\in\Gamma\backslash\{t,j\}}\dc_k < \dc_j$ is now equivalent to:
\begin{align}
&\sum_{k\in\Gamma\backslash\{t,j\}}\frac{\deps_t(1-c_t)}{D_k} < \frac{\deps_t(c_t-c_j)}{D_j-\deps_t} \\
\Rightarrow &\sum_{k\in\Gamma\backslash\{t,j\}}\frac{(1-c_t)}{D_k} < \frac{c_t-c_j}{D_j-\deps_t}
\end{align}
A sufficient condition to guarantee this is 
\begin{align}
\sum_{k\in\Gamma\backslash\{t,j\}}\frac{(1-c_t)}{D_k} \leq  \frac{c_t-c_j}{D_j}\Rightarrow \sum_{k\in\Gamma\backslash\{t,j\}}D_k \geq \frac{1-c_t}{c_t-c_j}D_j \label{constr:enough-coverage}
\end{align}

As a summary, as long as condition~\ref{constr:enough-cj} and \ref{constr:enough-coverage} is satisfied, we have a counterexample (when all manipulations are on reward of course). A concrete example is given below. Assume optimal attack target is $1$.
$ R^a_1 = R^a_2 = 2, P^a_1 = P^a_2 = -2. \forall k=3,\ldots, 9, R^a_k = 1.5, P^a_k = -8.5$. Budget $B=1$. It is easy to see $EU(k) = 1$ for the greedy manipulation to use budget to increase $R^a_1$ (increasing $P^a_1$ gives a worse solution). And $c_1 = \frac25, c_2 = \frac14, c_k =\frac{1}{20}$. One can verify both conditions can be satisfied.

\end{document}